\documentclass[a4paper,10pt, twocolumn]{IEEEtran}

\usepackage{graphicx}
\usepackage{subfigure}
\usepackage{cite}
\usepackage{multirow}
\usepackage{epsfig, psfrag, amsmath, amssymb, amsfonts, latexsym, amsthm}
\usepackage{bbm}
\usepackage{float}
\usepackage{geometry}
\newtheorem{thm}{Theorem}

\newtheorem{lem}{Lemma}

\begin{document}
\pagestyle{empty}
\title{Power Control in Full Duplex Networks:\\ Area Spectrum Efficiency and Energy Efficency}
\author{Chenyuan Feng$^*$, Yi Zhong$^*$, Tony Q.S. Quek$^*$ and Gang Wu$^\dagger$\\
$^*$Singapore University of Technology and Design, Singapore\\
$^\dagger$University of Electronic Science and Technology of China, Chengdu, China\\
Emails: chenyuan\_feng@mymail.sutd.edu.sg, \{yi\_zhong, tonyquek\}@sutd.edu.sg, wugang99@uestc.edu.cn
}
\maketitle
\begin{abstract}
 Full-duplex (FD) allows the exchange of data between nodes on the same temporal and spectrum resources, however, it introduces self interference (SI) and additional network interference compared to half-duplex (HD). Power control in the FD networks, which is seldom studied in the literature, is promising to mitigate the interference and improve the performance of the overall network. In this work, we investigate the random and deterministic power control strategies in the FD networks, namely, constant power control, uniform power control, fractional power control and ALOHA-like random on-off power control scheme. Based on the obtained coverage probabilities and their robust approximations, we show that power control provides remarkable gain in area spectrum efficiency (ASE) and energy efficiency (EE), and improves the fairness among the uplink (UL) and downlink (DL) transmissions with respect to the FD networks. Moreover, we evaluate the minimum SI cancellation capability to guarantee the performance of the cell-edge users in FD networks. Generally, power control is helpful to improve the performance of the transmission for long distance in the FD networks and reduce the requirement of SI cancellation capability.
\end{abstract}
\thispagestyle{empty}


\IEEEpeerreviewmaketitle
\section{Introduction}
Traditionally, radio transceivers are subject to a half duplex (HD) constraint because of the crosstalk between the transmit and receive chains. The self-interference (SI) caused by the transmitter at the receiver if using full duplex (FD) transmission overwhelms the desired received signal from the partner node. If the SI and increasing network interference introduced by FD can be well managed, FD communication can potentially double the spectrum efficiency. The throughput of FD network has been investigated, such as \cite{HFDHD} and \cite{FDgain}. Despite the large amount of literature on studying the area spectrum efficiency (ASE) of the FD networks, discussion on energy efficiency (EE) \cite{7054598, 7015548} has been rarely addressed. And whether FD communication can be used in large scale network is still an open question. Key to success is well designed power control scheme which can ensure the communication link reliability as well as mitigate network level interference.

In HD radio wireless networks, two main approaches have been used to analyze and design power control schemes, namely, the optimization approach and the game theory approach. The optimization approach takes a global point of view and aims at finding the assignment of power that maximizes certain global metrics \cite{RPCPN4}, \cite{RPCPN7}. These works model the users as cooperative individuals and attempt to propose distributed power control algorithms. The distributed algorithms with different global optimization goals were proposed in cellular networks \cite{RPCPN8} and ad hoc networks \cite{RPCPN4}, \cite{RPCPN7}, \cite{RPCPN9}. The game theory approach views the network as a collection of non-cooperative nodes with conflict interest and aims to determine solutions that maximize certain global utilities \cite{RPCPN10}, \cite{RPCPN17}. These works consider the cases where any inter-node coordination is not allowed and aim at finding robust power control schemes such that no individual node of the network can achieve a better expected performance by unilaterally deviating from current schemes \cite{RPCIPN}.

Most of the above works consider the power control strategies that the power level is deterministic at each transmitter. The work in \cite{RDPCIWN} introduces random power control with uniformly distributed transmit power, which reveals that uniformly distributed power control outperforms fixed power control in terms of outage probability. \cite{FPCFDWN} investigate the impact of the fractional power control on outage probability and transmission capacity in decentralized wireless networks. \cite{DPC} focuses on multi-user cases in a wireless clustered ad hoc network, and reveals that the discrete power control can significantly improve the transmission capacity and balance the UL and DL throughputs. \cite{RDPCIWN}, \cite{19} demonstrate the benefits of random power control in the presence of interference in HD networks.

Related works on power control in FD networks are listed as follows. By employing the Lagrange multiplier method, \cite{fdpc2} derived an optimal joint power control solution to maximize the sum rate of UL and DL transmissions. In \cite{fdpc3}, the optimal power control scheme was obtained via the exhaustive search for pairwise UE and BS power. However, most of the works related to power control in FD network focus on the single cell and cannot extend to the large scale networks directly. All the above power control schemes require instantaneous channel state information (CSI) at the transmitter, while is not required in our proposed power control schemes.

In this paper, we consider a large scale network with two-node bidirectional FD communication and develop a framework for FD networks in the presence of the SI and the network interference. Since the interference from the BSs is usually stronger than the interference from the UEs and it is more feasible for the BSs to implement power control, we only consider the power control in the DL, i.e., all BSs will employ the power controls and all UEs will transmit with fixed power. Then, we investigate the ASE and the EE of the network, and evaluate the tradeoff between UL and DL transmissions under different power control schemes.

This work reveals that the proposed power control schemes in the FD networks can provide remarkable gain for the throughput over the HD network and improve the fairness between the UL and DL transmissions. Additionally, this results show the benefit of randomness in the transmit power on the ASE and the EE, especially for the transmissions with short distance. For SI cancellation capability, the proposed power control schemes in FD networks help to improve the coverage of the network.

The remainder of the paper is organized as follows. Section \ref{sec:model} describes the system model and the power control schemes and then analyzes the property of the interference. Section \ref{sec:iapc} analyzes the coverage probability and the ASE and the EE for different power control schemes. Section \ref{sec:numerical results} quantifies the effect of different network parameters, such as self-interference cancellation capability, on the ASE and the EE. Finally, section \ref{sec:conclusion} concludes this paper.

\section{System Model}
\label{sec:model}
We consider a single-tier cellular network, where BSs and UEs are distributed randomly according to two independent homogeneous Poisson point processes (PPP) ${\Phi }$ and ${\Phi }_{\text{u}}$ with intensity $\lambda $ and ${\lambda }_{\text{u}}$, respectively. The UEs always access the nearest BS, which provides the largest long term receiving power; thus, the coverage region of each BS constructs a Voronoi cell (see Fig.\ref{fig:voronoicell}). We assume that the universal frequency reuse is adopted for all BSs over bandwidth $W$, and all users in the same cell adopt the time division multiple access to avoid intra-cell interference.
Each UE associates with its nearest BS with maximum average received power, thereby the link distances $R_i$ are Rayleigh distributed with mean $1/2\sqrt{\lambda}$, i.e., the probability density function (pdf) of the link distance $R_i$ is given by
\begin{equation}
\label{eq:linkdistance}
f_{R}\left(r \right) = \frac{d\mathbb{P}\left[\mathit{R}\leq r \right]}{dr}=2\pi \lambda r e^{-\lambda \pi r^2}, r \geq 0.
\end{equation}
According to the association rule, there will be some cells covering more than one UE, in which case the BS will randomly choose a UE to serve in each time slot and all users in the same cell adopt the time division multiple access (TDMA) to access the BS. Meanwhile, there might exist some BSs serving no UE, in which case these BSs will turn into sleeping mode to save energy and mitigate the network interference. The idle probability, denoted by $p_0$, could be obtained form \cite{idleprob} as
\begin{small}
\begin{equation}
 p_0 \approx {\left(1+{3.5}^{-1}\frac{\lambda_{\text{u}}}{\lambda} \right)}^{-3.5}.
\end{equation}
\end{small}

For analytical tractable, we approximate the locations of the active interfering BSs is an independently thinning version of the original PPP with probability $p_0$. In fact, this thinning is dependent on the cell size and the UE distribution.
In an arbitrary time slot, all active nodes are paired, and each pair consists of a single BS and its dedicate UE. Thereby, the network topology can be modeled by a marked Poisson point process (PPP) $\hat{\Phi}=\{\left(x_i,m(x_i)\right)\} \subset \mathbb{R}^2 \times \mathbb{R}^2$, where $\Phi_{\text{b}}=\{x_i\}$ is a homogeneous PPP with intensity $\lambda_{\text{b}}=\left(1-p_0\right)\lambda$ which denote the locations of BSs and the mark $m\left(x_i\right)$ denotes the location of its dedicated UE. The link distance, denoted by $R_{x_i}={\|x_i-m(x_i)\|}$, and $\{R_i\}$ are iid random variables with pdf $f_R$ given by (\ref{eq:linkdistance}).

All UEs and BSs are equipped with a single antenna and work in the FD mode. The SI is assumed to be cancelled imperfectly with residual SI-to-power ratio $\beta$, i.e., when the transmit power of a node is $P_{\text{r}}$, the residual SI is $\beta P_{\text{r}}$. The parameter $\beta$ quantifies the amount of SI cancellation capability. When $\beta= 0$, the SI cancellation is perfect, and for $\beta= 1$, there is no SI cancellation.

We consider a typical BS at $x_0$ and its associated UE at $m(x_0)$. Without loss of generality, we assume that the typical UE is at the origin, i.e., $m(x_0)=\mathit{o}$. With a propagation channel model with path loss, $l(r)=r^{-\alpha}, \alpha >2 $, where $\alpha$ is the path loss coefficient and Rayleigh fading, the signal-to-interference ratio (SIR) is
\begin{small}
\begin{equation}
\mathrm{SIR}=\frac{P_{\text{t}} h R^{-\alpha}}{P_{\text{r}}\beta+ I}, \end{equation}
\end{small}
\hspace*{-0.4em}where $P_{\text{t}}$ and $P_{\text{r}}$ are the transmit power at the transmitter and the receiver respectively, and $h$ is power fading coefficients from the transmitter to the desired receiver. We focus on the i.i.d Rayleigh fading case, $h$ is exponentially distributed with unit mean. Let $\theta_{\text{r}}$ be the threshold of the SIR for successful transmission, i.e., a transmission is successful if the SIR is greater than $\theta_{\text{r}}$.

The aggregate interference at the typical UE is given by
\begin{small}
\begin{equation}
\label{eq:interference}
I=\sum \limits_{x_i\in \Phi_{\text{b}}\backslash\{x_0\}} P_{x_i} h_{x_i} d^{-\alpha}_{x_i} + P_{\text{u}} h_{m\left(x_i \right)} d^{-\alpha}_{m\left( x_i\right)},
\end{equation}
\end{small}
\hspace*{-0.5em}where $ d^{-\alpha}_{{x}_{i}}={\|x_i-m(x_0)\|}^{-\alpha }$ is the path loss in the interfering link, ${\| \cdot \|}$ denotes the Euclidean norm, $h_{x_i}$ and $h_{m(x_i)}$ are independent and identically distributed (i.i.d.) exponentially distributed random variables with unit mean. We assume that the UEs transmit with constant power $P_{\text{u}}$ and the BS located at $x_i$ transmits with power $P_{x_i}$, the values of which are determined by the power control schemes.

And in this paper, we will discuss several common power control scheme used in HD network in order to find whether these power control schemes benefit the FD network. For lack of space, we don't include the part about how to tune the parameters of each power control scheme so as to extend these schemes and related equipment to FD networks in this paper. For comparison, all power control schemes are under a same peak power constraint, denoted by $P_{\max}$.

\setcounter{equation}{9}
\begin{figure*}[!ht]
\begin{small}
\begin{equation}
\label{eq:fpcfp}
f_P(x)=\frac{2\pi\lambda x^{\frac{2}{\alpha\epsilon}-1}e^{-\pi\lambda{(\frac{x}{\bar{P}})}^{\frac{2}{\alpha\epsilon}}}}{\alpha\epsilon{\bar{P}}^{\frac{2}{\alpha\epsilon}}}+e^{-\pi\lambda{(\frac{P_{\max}}{\bar{P}})}^{\frac{2}{\alpha\epsilon}}}\delta(x-P_{\max}),  P_{\min}\leq x\leq P_{\max},
\end{equation}
\end{small}
\hrulefill
\setcounter{equation}{10}
\end{figure*}

\begin{itemize}
    \item{Constant Power Control (CPC)}

Under the CPC scheme, all BSs transmit with  peak power.
The transmit power of all BSs and its pdf are
\setcounter{equation}{4}
\begin{align}
 &P_{x_i}=P_{\max},\\
 &f_P(x)=\delta(x-P_{\max}),
\end{align}
\setcounter{equation}{6}
\hspace*{-0.4em}where $\delta(\cdot)$ is Dirac delta function.
\item{Uniform Power Control (UPC)}

Under the UPC policy, all BSs will independently and randomly choose their transmit power from the range $\left[P_{\min},P_{\max}\right]$. The values of $P_{\min}$ and $P_{\max}$ are the same for all BSs. Then the transmit power of the BSs located at $x_i$ and its pdf are
{
\begin{align}
&P_{x_i} \sim \text{Uniform} \left[P_{\min}, P_{\max}\right],\\
&f_P(x)=\frac{1}{P_{\max}-P_{\min}}, P_{\min} \leq x \leq P_{\max}.
\end{align}}

\item{Fractional Power Control (FPC)}

Under the FPC scheme, the transmit power of each BS is tuned according to the path loss of the desired link; specially, it is chosen as the path loss raised to an exponent, denoted by $\epsilon$, where $\epsilon$ is a constant within $\left[0,1\right]$. The FPC scheme depends on the link distance between the BS and its dedicated UE, thus, this scheme is helpful to compensate the path loss and overcome the near-far problem. Then the transmit power of the BSs located at $x_i$ is
\begin{equation}
P_{x_i}=\begin{cases}
 & \bar{P}R^{\alpha\epsilon}_{x_i},\text{ if } \bar{P}R_{x_i}^{\alpha\epsilon}\leq P_{\max},\\
 & P_{\max}, \text{ if } \bar{P}R_{x_i}^{\alpha\epsilon} > P_{\max},
\end{cases}
\end{equation}
where $\bar{P}$ is a constant, which is a cell-specific parameter and its pdf is given as (\ref{eq:fpcfp}).

\item{Aloha-like Random Power Control (APC)}

Similar to ALOHA protocol, under the ALOHA-like random on-off power control scheme, each BS will transmit at a certain power level with certain probability; otherwise, it will turn into sleeping mode. Both the transmit power and the transmit probability are chosen based on the desired link distance and the SIR threshold. This scheme improves the link reliability but reduces bandwidth efficiency since the BS only transmits when channel condition is good.
Then the transmit power of the BSs located at $x_i$ and its pdf are
\setcounter{equation}{10}
\begin{align}
&P_{x_i}=\begin{cases}
 & \bar{P},\text{ at probability } \xi,\\
 & 0, \text{at probability } 1-\xi,\\
\end{cases}\\
\label{eq:fpapc}
& f_P(x)=\xi\delta(x-\bar{P}),
\end{align}
\setcounter{equation}{12}
\hspace*{-0.5em}where $\delta(\cdot)$ is Dirac delta function, $\bar{P}$ is random transmit power level, $\xi$ is transmit probability.
\end{itemize}

\section{Network Analysis}
\label{sec:iapc}
\subsection{Coverage Probability}
Based on the aforementioned SIR model given by (\ref{eq:interference}), the coverage probability at the typical receiver is
\begin{align}
\label{eq:coverage}
    p_{\text{c,r}}&=\mathbb{P}\left(\frac{P_{\text{t}} h R^{-\alpha}}{P_{\text{r}}\beta +I}>\theta_r\right)\nonumber\\
    &=\mathbb{E}_{P,R,I}\left\{\mathbb{P}\left(h>\frac{\theta_{\text{r}} (\beta P_{\text{r}} +I) }{P_{\text{t}}}R^{\alpha}\bigg|P_{\text{t}},\mathit{R}\right)\right\}\nonumber\\
    &= \mathbb{E}_{P,R,I}\left\{ \exp\left(-\frac{\theta_{\text{r}} R^{\alpha}}{P_\text{t}}P_{\text{r}}\beta\right) \exp\left(-\frac{\theta_{\text{r}}}{P_{\text{t}}} R^{\alpha}I\right)\right\}\nonumber\\
    &=\mathbb{E}_{P,R}\left\{ \exp\left(-sP_{\text{r}}\beta\right)\mathcal{L}_I(s)|_{s=\frac{\theta_{\text{r}} R^{\alpha}}{P_{\text{t}}}}\right\},
\end{align}
where $P_{\text{t}}$, $P_{\text{r}}$ is transmission power at transmitter and receiver, respectively, $\theta_{\text{r}}$ is SIR target value. Notice that (\ref{eq:coverage}) is applicable to both UL and DL transmissions. 
\begin{proof}
step (a) is due to the total probability formula, step (b) is because of $|h| \sim \exp(1)$, and step (c) is for Laplace transform.
\end{proof}

\begin{figure*}[!ht]
\begin{small}
\setcounter{equation}{21}
\begin{align}
     \label{eq:capacity}
    \tau^{\text{FD}}&=R_{\text{u}}\mathbb{P}\left(\frac{Ph_{x_0}R^{-\alpha}}{P_{\text{u}}\beta+I}>\theta_{\text{u}}\right)+R_{\text{b}}\mathbb{P}\left(\frac{P_{\text{u}}{h}_{m(x_0)}R^{-\alpha}}{P\beta+I}>\theta_{\text{b}}\right)\nonumber\\
    &=R_{\text{u}}\mathbb{E}_{P,R}\left\{e^{-\frac{\theta_{\text{u}} R^{\alpha}}{P}P_{\text{u}}\beta}\mathcal{L}_I(\frac{\theta_{\text{u}} R^{\alpha}}{P})\right\}+R_{\text{b}}\mathbb{E}_{P,R}\left\{e^{-\frac{\theta_{\text{b}} R^{\alpha}}{P_{\text{u}}}P\beta}\mathcal{L}_I(\frac{\theta_{\text{b}} R^{\alpha}}{P_{\text{u}}})\right\}.\\
\setcounter{equation}{22}
\setcounter{equation}{25}
\label{eq:EE}
\eta_{\text{EE}}
&=\mathbb{E}_{\text{P,R}}\left\{\frac{R_{\text{u}}e^{-\frac{\theta_{\text{u}} R^{\alpha}}{P}P_{\text{u}}\beta}\mathcal{L}_I(\frac{\theta_{\text{u}} R^{\alpha}}{P})+R_{\text{b}}e^{-\frac{\theta_{\text{b}} R^{\alpha}}{P_{\text{u}}}P\beta}\mathcal{L}_I(\frac{\theta_{\text{b}} R^{\alpha}}{P_{\text{u}}})}{P+P_{\text{u}}+P_{\text{s}}} \right\}.
\end{align}
\setcounter{equation}{26}
\end{small}
\hrulefill
\end{figure*}

\begin{thm}
\label{thm:laplace}
 The Laplace transform of the aggregate interference at the typical UE is
\setcounter{equation}{13}
\begin{small}
\begin{multline}
\label{eq:laplace}
     \mathcal{L}_{I}\left(s\right)=\exp\bigg(-2\pi\lambda_{\text{b}} \mathbb{E}_P\bigg\{\int_{0}^{\infty}\bigg(1- \int_{0}^{2\pi}\int_{0}^{\infty}\frac{1}{1+sPv^{-\alpha}}\\
     \times \frac{\lambda r e^{-\pi\lambda r^2}drd\varphi}{1+su\left(v^2+r^2+2vr\cos\varphi \right)^{-\alpha/2}}\bigg)vdv\bigg\}\bigg).
\end{multline}
\end{small}
\setcounter{equation}{14}
\end{thm}

\begin{proof}
For geometric knowledge, we have $d_{m(x_i)}= d_{x_i}^2 + R_{x_i}^2-2R_{x_i}d_{x_i}\cos(\varphi)$, using Laplace functional, we can obtain

\scalebox{0.85}{
\parbox{0.54\textwidth}{
\begin{align}
\label{eq:laststep}
\mathcal{L}_{I}\left(s\right)
 &=\mathbb{E}\left\{e^{-\sum \limits_{x_i\in \Phi_{\text{b}}\backslash\{x_0\}} s\left(P_{x_i} h_{x_i} d^{-\alpha}_{x_i} + P_{\text{u}} h_{m\left(x_i \right)} d_{m\left( x_i\right)}^{-\alpha} \right)}\right\}\notag\\
 &=\mathbb{E}_{\Phi _{\text{b}},P,h}\left\{\prod_{x_i\in \Phi_{\text{b}}\backslash\{x_0\}} e^{ -s(P_{x_i} h_{x_i}d_{x_i}^{-\alpha }+P_{\text{u}} h_{m(x_i)} d_{m(x_i)}^{-\alpha })}\right\}\notag\\
 &\overset{\text{(a)}}{=}\mathbb{E}_{\Phi_{\text{b}},P}\left\{\prod_{x_i\in \Phi_{\text{b}}\backslash\{x_0\}}\frac{1}{1+sP_{x_i}d_{x_i}^{-\alpha}}\frac{1}{1+sP_{\text{u}}d_{m\left(x_i\right)}^{-\alpha}}\right\},
\end{align}
}
}
where (a) is obtained by taking expectation of $h$ and the deduction from (\ref{eq:laststep}) to (\ref{eq:laplace}) follows from the probability generating functional (PGF) of the PPP.
\end{proof}
The Laplace transform is not in closed-form. However, we can further decouple the UL and the DL calculation and simplify the result by deriving the upper and lower bounds for the Laplace transform of the interference.

Assuming the spatial distribution of the BSs and the UEs as a compound PPP, the locations of BSs and their dedicate UEs are approximated to be co-located, i.e, $d_{m(x_i)}={(d_{x_i}^2+R_{x_i}^2-2R_{x_i}d_{x_i})}^{\frac{1}{2}} \approx d_{x_i}$, we obtain the following Lemma.
\begin{lem}
\label{thm:approximation1}
The Laplace transform of the interference is upper bounded by
\begin{small}
\begin{equation}
     \label{eq:approximation}
     {\mathcal{L}}_{I}(s) \leq  \exp\left(-\lambda_{\text{b}}\frac{{\pi}^{2}\delta}{\sin\left(\pi\delta\right)}{s}^{\delta}\mathbb{E}_P\left\{\frac{{P_{\text{u}}}^{1+\delta}-\mathit{P}^{1+\delta}}{P_{\text{u}}-\mathit{P}}\right\}\right).
\end{equation}
\end{small}
\hspace*{-0.5em}where $\delta=2/\alpha <1$.
\end{lem}
\begin{proof}
See Appendix \ref{ap:app1}.
\end{proof}

Assuming the spatial distribution of the BSs and the UEs as two independent homogeneous PPPs, we obtain the following Lemma.
\begin{lem}
\label{thm:approximation2}
The Laplace transform of the interference is lower bounded by
\begin{small}
\begin{equation}
\label{eq:approximation2}
  {\mathcal{L}}_{I}(s) \geq  \exp\left(-\lambda_{\text{b}}\frac{{\pi}^2 \delta}{\sin\left(\pi\delta \right)}s^{\delta}\left(P_{\text{u}}^{\delta}+\mathbb{E}_P\{P^{\delta}\} \right) \right),
\end{equation}
\end{small}
\hspace*{-0.5em}where $\delta=2/\alpha <1$.
\end{lem}
\begin{proof}
See Appendix \ref{ap:app2}.
\end{proof}
From the above lemmas, we obtain upper and lower bounds for the Laplace transform of the interference which are in the same form of
\begin{equation}
\label{eq.uniformexpression}
\mathcal{L}_{\mathit{I}}(s) = \exp\left(- \lambda_{\text{b}}\frac{\pi^2\delta}{\sin(\pi\delta)}g(\delta,P){\mathit{s}}^{\delta}\right),
\end{equation}
where $g(\delta,P)$ is a function of the $\delta$-th moment of transmit power $P$ with the pdf $f_P(\cdot)$, which is:
\begin{equation}
    \label{eq.picewiseofa}
    g(\delta,P)=\left\{\begin{matrix}
 \mathbb{E}_P\left\{\frac{P_{\text{u}}^{1+\delta}-P^{1+\delta}}{P_{\text{u}}-P}\right\},&\text{upper bound}\\
 \mathbb{E}_P\left\{P_{\text{u}}^{\delta}+P^{\delta}\right\} ,&\text{lower bound}
\end{matrix}\right.
\end{equation}
Since $\delta=2/\alpha<1$ and $P_{\text{u}} \not\equiv P$, $g(\delta,P)>0$ always holds. Thus $\mathcal{L}_{\mathit{I}}(0)=1$, $\lim \limits_{s\rightarrow \infty}\mathcal{L}_{\mathit{I}}(\mathit{s})=0$, and $\mathcal{L}_{\mathit{I}}(\mathit{s})$ is monotonically decreasing on $\left[0,\infty \right)$.

The tightness of these two bounds can be measured by the cross correlation of two point process $\Phi_{\text{b}}$ and $\Phi_{\text{u}}$. For lack of space, this proof is not mentioned in this paper. The closed-form bounds of $\mathcal{L}_I(s)$ obtained is useful to explore the optimal parameters for the proposed power control schemes.

In Full Duplex radio, the coverage probabilities for UL and DL are
\begin{small}
\begin{equation}
p_{\text{ul}}=\mathbb{E}_{P,R}\left\{\exp\left(-\theta_{\text{b}} R^{\alpha}\beta\frac{P}{P_{\text{u}}}-\frac{\lambda_{\text{b}}\pi^2\delta\theta_{\text{b}}^{\delta} R^{2}}{\sin(\pi\delta)} \frac{g(\delta,P)}{P_{\text{u}}}\right)\right\},
\end{equation}
\end{small}
\begin{small}
\begin{equation}
  p_{\text{dl}}=\mathbb{E}_{P,R}\left\{\exp\left(-\theta_{\text{u}} R^{\alpha}\beta\frac{P_{\text{u}}}{P}-\frac{\lambda_{\text{b}}\pi^2\delta\theta_{\text{u}}^{\delta} R^{2}}{\sin(\pi\delta)} \frac{g(\delta,P)}{P}\right)\right\}.
\end{equation}
\end{small}
The inequalities $\frac{\partial p_{\text{c,dl}}}{\partial{P}} > 0$ and $\frac{\partial p_{\text{c,ul}}}{\partial{P}} < 0$ hold for all parameters, but the equations $\frac{\partial^2 p_{\text{c,dl}}}{\partial P^2} = 0$  and $\frac{\partial^2 p_{\text{c,ul}}}{\partial P^2} = 0$ are transcendental equations, whose analytical solutions are hard to derive.

\subsection{Area Spectrum Efficiency}
The target data rates of the BSs and the UEs may be different, denoted by $R_{\text{b}}$ and $R_{\text{u}}$, and the SIR thresholds for DL and UL are $\theta_{\text{b}}=2^{R_{\text{b}}/W}-1$ and $\theta_{\text{u}}=2^{R_{\text{u}}/W}-1$, respectively, where $W$ is the spectrum bandwidth.
The achievable data rates are $\tau_{\text{m,b}}=R_{\text{b}} p_{\text{m,ul}}$ and $\tau_{\text{m,u}}=R_{\text{u}} p_{\text{m,dl}}$, $m\in\{\text{CPC,FPC,UPC,APC}\}$.
For a FD link, the sum data rate is given as (\ref{eq:capacity}).

In HD networks, the UL and DL achievable data rates are
\setcounter{equation}{22}
\begin{align}
 \tau_{\text{HD,b}}&=0.5W\log_2(1+\theta_{\text{b}})\mathbb{E}_R\left\{e^{-\frac{\lambda_{\text{b}}\pi^2\delta}{\sin(\pi\delta)}{(\theta_{\text{b}} R^{\alpha})}^{\delta}}\right\},\label{eq:HDUL}\\
 \tau_{\text{HD,u}}&=0.5W\log_2(1+\theta_{\text{u}})\mathbb{E}_{R}\left\{e^{-\frac{\lambda_{\text{b}}\pi^2\delta}{\sin(\pi\delta)}{(\theta_{\text{u}} R^{\alpha})}^{\delta}}\right\}.\label{eq:HDDL}
\end{align}
\setcounter{equation}{24}
ASE with unit $\text{bps}/\text{Hz}/\text{m}^2$ is given by
\begin{align}
\eta_{\text{m,ASE}}&=\frac{\lambda_{\text{b}}\left(\tau_{m,\text{u}}+\tau_{m,\text{b}}\right)}{W}\nonumber\\
&=\lambda_{\text{b}} \left[\log_2\left(1+\theta_{\text{u}} \right)p_{\text{m,dl}}+\log_2\left(1+\theta_{\text{b}}\right)p_{\text{m,ul}}\right].
\end{align}
\subsection{Energy Efficiency}
EE with unit $\text{bps}/\text{J}$ is given by (\ref{eq:EE}), where $P_{\text{s}}$ is static power consumption while communicating.

Since the power control is only applied in the DL, a potential advantage is that the DL and UL performance could be balanced. For a given $\beta$, if BS employs full power, the UL coverage is bad while the DL coverage is improved. In this regard, we will maximize the minimum achievable data rate by optimizing the parameters of different power control schemes. The comparison of the performance for different power control schemes are shown in the numerical results.

\section{Numerical Results}
\label{sec:numerical results}

In this section, we illustrate the performance of different power control schemes through numerical evaluations. We also evaluate the effect of the system parameters on the tightness of the upper and lower bound. In particular, we will investigate the coverage probabilities, the area spectrum efficiency and the energy efficiency based on the numerical results. The default system parameters are shown in Table \ref{tb:1}.

\begin{table}[tbp]
\centering
\caption{System parameters}
\label{tb:1}
\begin{tabular}{lccc}
\hline
Symbol & Description & Value \\
\hline
  $\lambda$            & \text{Density of BS}   & $10^{-6}$ /$\text{m}^{-2}$\\
  $\lambda_{\text{u}}$ & \text{Density of UL}   & $10^{-5}$ /$\text{m}^{-2}$\\
  $\theta_{\text{b}}$  &\text{ UL SIR threshold } &0\text{dB}\\
  $\theta_{\text{u}}$  &\text{ DL SIR threshold}  & 0\text{dB}\\
  $\alpha$             & \text{Path loss exponent  }    & 4\\
  $\beta$              & \text{Residual SI-to-power ratio} & -100 \text{dB}\\
  $P_{\text{u}}$       & \text{Transmit power of UE }   & 23 \text{dBm} (200\text{ mW})\\
  $P_{\text{s}}$       & \text{Static power consumption} & 150 \text{mW}\\
  $P_{\max}$           &\text{ Peak transmit power of BS} & 43 \text{dBm} (2 \text{W})\\
  $P_{\min}$           & \text{BS minimum transmit power} & 200 \text{mW}\\
  $\epsilon$           & \text{Fractional power control exponent} & 0.1\\
  $W$                  & \text{Spectrum width}       & 10 \text{MHz}\\
  $R_{\text{b}}$       & \text{Target UL data rate}       &　10 \text{Mbps}\\
  $R_{\text{u}}$　　　&\text{Target DL data rate}　　　&　10 \text{Mbps}\\
\hline
\end{tabular}
\end{table}

\subsection{Tightness of Bounds}

Two different approaches are proposed to derive the upper and lower bound for the coverage probabilities in Lemma \ref{thm:approximation1} and Lemma \ref{thm:approximation2}. The lower bound is obtained by the independent approximation, and the upper bound is obtained by assuming the link distances from the interfering BS and the distance from the dedicated UE are linearly correlated.
The tightness of these bounds depends on the system parameters. We consider four curves: the actual coverage probabilities determined by Monte-Carlo simulation, the numerical evaluations of the derived coverage probabilities, the lower bound and the upper bound of the coverage probabilities. The tightness of the bounds is shown in Fig. \ref{pic.tightness}.

From Fig. \ref{pic.tightness}, we observe that both the analytical expressions and the bounds approach the simulation results well, and numerically the lower bound is a tight approximation. Since the tightness of the bounds highly depends on the values of $sP_{x_i}d_{x_i}^{-\alpha}$ and $sP_{\text{u}}d_{m(x_i)}^{-\alpha}$, we plot the Laplace transform of the interference for different node densities $\lambda_{\text{b}}$, different peak transmit powers $P_{\max}$, and different path loss coefficients $\alpha$.

\subsection{SI Cancellation Requirement}

Fig. \ref{fig:a} implicates that power control schemes greatly improve the performance of FD communication when the distance is large. In other word, implementing the proposed power control schemes notably improves the performance for the cells with large coverage region in the FD networks. For the default network parameters, FPC outperforms others in terms of the ASE. For example, for $\beta=-80$ dB, FD network with CPC outperforms HD network only when the link distance is less than 180 m, and FD network with FPC outperforms HD when the link distance is less than 250m. For $\beta=-100$ dB, FD network with CPC outperforms HD network when the link distance is less than 330 m, and FD network with FPC outperforms HD when the link distance is less than 500m. The performance of the FD networks decreases when increasing the link distance. Without power control schemes, the FD outperforms the HD network which requires perfect SI cancellation capability and cells with small coverage region.
Implementing power control schemes alleviates the requirement for the SI cancellation capability.

\subsection{ASE and EE}
In Fig. \ref{fig:ASEEEinPamax}, when the peak power constraint $P_{\max}$ increases, the EE decreases and the ASE performs differently for different power control schemes. For CPC and APC schemes, the ASE decreases as $P_{\max}$ increases, because when $P_{\max}$ increases, the network interference increases rapidly. For CPC scheme, both the SIR and the achievable rate decreases, then the ASE and the EE decrease. For APC scheme, each BS will choose low transmit probability to guarantee the link quality, the ASE decreases and the EE decreases as a result.


As Fig. \ref{fig:tradeoff} shows, for the APC scheme and the CPC scheme, the DL performance dominates the FD performance all the time. When the ASE increases, the ratio of the ASE of DL to the ASE of UL decreases to 1, which means that the ASE as well as the fairness can be optimized at the same time.

For FPC, the DL performance dominates the FD performance in terms of ASE all the time. When the ASE increases, the ratio of the ASE of DL to the ASE of UL increases to 1. For UPC scheme, the variation of the ASE is small and the offload between UL and DL can be flexible.

For APC scheme, the EE increases as the ratio of the EE of DL to the EE of UL decreases to 1, the EE and the fairness between DL and UL can be optimized at the same time. For CPC scheme, the EE increases as the ratio between DL and DL increases and the DL dominates the performance all the time, which means the EE for CPC scheme increases at the cost of the fairness between UL and DL.

For UPC and FPC, the EE increases as the ratio of the EE of DL to the EE of UL decreases, the offload between UL and DL can be flexible. Fig. \ref{fig:tradeoff} indicates that the APC scheme can be viewed as a fair power control scheme since both the EE and ASE approach their maximum values when the ratios between the ASE and the EE of UL and the ASE and the EE of DL equals to 1.

As for other power control schemes, increasing the ASE or the EE is achieved at the cost of decreasing the fairness between the DL and the UL, this unfairness can be used when UL and DL has different offload. For example, when the traffic demand of DL is twice of the traffic demand of UL, the ASEs of the FPC scheme, the APC scheme and the UPC scheme are 0.45bps/Hz/$\text{km}^2$, 0.5bps/Hz/$\text{km}^2$ and 0.56bps/Hz/$\text{km}^2$ respectively, then UPC is a better scheme in terms of the ASE. A novel hybrid power control schemes for FD networks need to be explored to further improve the performance of the FD networks in the large-scale networks.

\section{Conclusion}
\label{sec:conclusion}
In this paper, we proposed two approaches to approximate and bound the Laplace transform of the interference and prove their tightness. These closed-form results are useful when obtaining optimal parameters for the power control schemes. Furthermore, we evaluate the achievable data rates, the ASE and the EE of FD radio networks for different power control schemes.
%

Numerical results reveal that the ASE is improved by the randomness in the transmit power. The FD network outperforms the HD network when high SI cancellation capability and short link distance (small cells) are implemented. Given SI cancellation capability, the proposed power control schemes greatly improve the performance of the transmission for long distance in the FD networks. The impact of SI cancellation capability and peak power constraint on the performance of the proposed power control schemes are also investigated. These results can be used to design a dynamically manageable power control schemes in large scale wireless network to achieve significant energy savings while guarantee the quality of service of the UEs.

\appendices
\section{Proof of Lemma \ref{thm:approximation1}}
\label{ap:app1}
From the Fortuin-Kasteleyn-Ginibre (FKG) inequality \\(\cite{15}, Thm.10.13), we can get the following inequality:
\setcounter{equation}{26}
\begin{align}
\mathcal{L}_{I}(s)
&\geq \mathbb{E}_{\Phi_{\text{b}},P,h}\left\{\prod_{x_i\in \Phi_{\text{b}}\backslash\{x_0\}} e^{ -s P_{x_i} h_{x_i}d_{x_i}^{-\alpha }}\right\}\nonumber\\
&\times \mathbb{E}_{\Phi_{\text{b}},h}\left\{\prod_{x_i\in \Phi_{\text{b}}\backslash\{x_0\}} e^{-sP_{\text{u}} h_{m(x_i)} d_{m(x_i)}^{-\alpha }}.\right\}
\end{align}
\setcounter{equation}{27}
And the first term at the right side of inequality can be calculated as
\begin{small}
\begin{equation}
   {\mathbb{E}}_{\Phi _{\text{b}},P,h}\left\{\prod_{x_i\in \Phi_{\text{b}}\backslash\{x_0\}} e^{ -s P_{x_i} h_{x_i}d_{x_i}^{-\alpha }}\right\}=e^{-\frac{\lambda_{\text{b}}\pi^2\delta}{\sin\left(\pi\delta\right)}s^{\delta}\mathbb{E}_P\left\{P^{\delta} \right\}}.
\end{equation}
\end{small}
While the second term at the right side of the inequality can be obtained by same way due to the displacement theorem \\(\cite{15}). Multiplying two term completes the proof.

\section{Proof of Lemma \ref{thm:approximation2}}
\label{ap:app2}
When assuming $\Phi$ and $\Phi_{\text{u}}$ as a compound PPP, i.e., $d_{x_i} = d_{m{x_i}}$, the channel fading can be denoted by $G_i=P_{x_i}h_{x_i}+P_{\text{u}}h_{m\left(x_i\right)}$. Conditioned on given $P$, with similar proof in \cite{HFDHD}, we have
\begin{small}
\begin{equation}
\mathbb{E}_{G_i}\{G_i^\delta\}=\mathbb{E}_P\left\{\frac{\Gamma \left(1+\delta  \right)\left(P_{\text{u}}^{1+\delta } -P^{1+\delta }\right)}{P_{\text{u}}-P} \right\}.
\end{equation}
\end{small}
Then we obtain the upper bound as
\begin{align}
    \mathcal{L}_{I}\left(s \right)& \leq \exp\left(-\pi\lambda_{\text{b}}\mathbb{E}_{G_i}\left\{ {G_i}^\delta\right\}\Gamma \left(1-\delta  \right)s^{\delta } \right)\nonumber\\
   &=\exp\left(-\pi\lambda_{\text{b}}\mathbb{E}_P\left\{\frac{\pi\delta}{\sin\left(\pi\delta\right)}\frac{s^{\delta}\left(P_{\text{u}}^{1+\delta}-P^{1+\delta}\right)}{P_{\text{u}}-P}\right\}\right)\nonumber\\
   &=\exp\left(-\frac{\lambda_{\text{b}}{\pi}^{2}\delta{s}^{\delta}}{\sin\left(\pi\delta\right)}\mathbb{E}_P\left\{\frac{P_{\text{u}}^{1+\delta}-P^{1+\delta}}{P_{\text{u}}-P}\right\}\right).
\end{align}

\ifCLASSOPTIONcaptionsoff
  \newpage
\fi

\bibliographystyle{IEEEtran}
\bibliography{reference}

\end{document}